\documentclass[showpacs, showkeys,twocolumn,preprintnumbers,floatfix,amsmath,amssymb]{revtex4}
\usepackage{booktabs}
\usepackage{natbib}
\usepackage{amsmath,amsthm}
\usepackage{float}
\usepackage{graphicx}% Include figure files
\usepackage{epstopdf}
\usepackage{caption}
\usepackage{subfigure}
\usepackage{tabularx}% Include figure files
\usepackage{makecell}
\usepackage{threeparttable}
\usepackage{dcolumn}% Align table columns on decimal point
\usepackage{bm}% bold math
\usepackage{color,epsfig,multirow}
\usepackage{array}
\usepackage{hyperref}
\usepackage{algorithm}
\usepackage{algorithmic}
\usepackage{threeparttable}

\numberwithin{equation}{section}

\newcommand{\cO}{\mathcal{O}}

\newtheorem{theorem}{Theorem}[section]

\theoremstyle{definition}

\theoremstyle{remark}

\allowdisplaybreaks[4]

\begin{document}
\preprint{Preprint}

\title{Variance-reduced random batch Langevin dynamics}
\author{Zhenli Xu$^{1,2}$}\thanks{xuzl@sjtu.edu.cn}
\author{Yue Zhao$^3$}\thanks{zhaoyu14@msu.edu}
\author{Qi Zhou$^{1}$}\thanks{zhouqi1729@sjtu.edu.cn}
\affiliation{$^1$School of Mathematical Sciences, Shanghai Jiao Tong University, Shanghai 200240, China\\
$^2$CMA-Shanghai, MOE-LSC and Shanghai Center for Applied Mathematics, Shanghai Jiao Tong University, Shanghai 200240, China \\ $^3$Department of Computational Mathematics,
Science \& Engineering, Michigan State University, East Lansing,
MI 48824, USA}
%Lines break automatically or can be forced with \\
\date{\today}
%%%%% Begin Abstract %%%%%%%%%%%
\begin{abstract}
The random batch method is advantageous in accelerating force calculations in particle simulations, but it poses a challenge of removing the artificial heating effect in application to the Langevin dynamics.
We develop an approach to solve this issue by estimating the force variance, resulting in a variance-reduced random batch Langevin dynamics.
Theoretical analysis shows the high-order local truncation error of the time step in the numerical discretization scheme, in consistent with the fluctuation-dissipation theorem.
Numerical results indicate that the method can achieve a significant variance reduction since a smaller batch size provides accurate approximation, demonstrating the attractive feature of the variance-reduced random batch method for Langevin dynamics. 
\end{abstract}

%%%%% end %%%%%%%%%%%

%%%%% AMS/PACs/Keywords %%%%%%%%%%%

\pacs{02.70.Ns, 87.15.Aa}
%02.70.Ns	Molecular dynamics and particle methods
%87.15.Aa	Theory and modeling; computer simulation
%33.15.Dj	Interatomic distances and angles
%64.70.Fx	Liquid-vapor transitions
%\ams{}
\keywords{Langevin dynamics, random batch list method, variance reduction.}

%%%% maketitle %%%%%
\maketitle

 %%%% Start %%%%%%
\section{Introduction}

Langevin dynamics simulation \cite{frenkel2001understanding} is a highly influential computational technique utilized in various scientific disciplines, including physics, chemistry, biology, and material science, specifically at the nano- and micro-scales \cite{mabey2017strong,Trullas1989Langevin,Lee2017computational,Pastor1994}. Its primary purpose is to investigate the equilibrium properties and dynamics of particles in many-body systems.
As a thermostat in molecular dynamics simulation, this computational approach simulates canonical ensembles by integrating Newton's equation of motion combined with thermal fluctuations and frictions \cite{welling2011bayesian,ma2015complete,roberts1996,dalalyan2017,oksendal03}, and its widespread applications have been found in the analysis of irreversible dynamics, stochastic properties, and trajectory instability of various systems \cite{norman2013stochastic}. Furthermore, in cutting-edge fields such as machine learning, Langevin dynamics and its improved algorithms can efficiently transform stochastic optimization into posterior sampling, demonstrating strong potential for a wide range of learning tasks \cite{welling2011bayesian,ma2015complete,dalalyan2017,dubey2016variance,coeurdoux2024normalizing}.

The computational bottleneck for particle simulations using Langevin dynamics is non-bonded interactions between particles. The recently developed random batch method (RBM) \cite{jin2018random} shows promise in the fast calculation of interacting particle systems. This method is inspired by the mini-batch idea in stochastic gradient descent algorithms of machine learning \cite{robbins1951stochastic,bottou1998online,bubeck2015convex}, and approximates the force by using a batch of randomly selected neighbors of a fixed batch size such that the approximate force is unbiased, achieving much success \cite{CiCP-28-1907,li2020Random,jin2022mean,li2020some,fornasier2022anisotropic,ko2021model,golse2019random,ye2024error,guillin2024some}. Particularly, in molecular dynamics, the RBM has been extended to obtain the random batch Ewald method and the random batch sum-of-Gaussians method, which solve
the parallel scalability for long-range interactions  \cite{JinLiXuZhao2020,liang2022random,liang2023random,gan2024random}. For short-range interactions, the singular kernel can be treated by the random batch list (RBL) method \cite{liang2021randomRBL,liang2022IRBE,lin2024hybrid}.
The RBL divides the interaction region into a core-shell structure and applies the RBM in the shell region, significantly reducing the neighboring number of particle interactions. The RBM is implemented for molecular dynamics at the graphic processing units and allows the simulation of a system of ten million particles in a single card \cite{gao2024rbmd}.

The randomness introduced into the deterministic drifting term results in an additional error of $\cO(\tau)$ at each discrete time step \cite{liang2021randomRBL}, leading to the failure of the fluctuation-dissipation theorem \cite{kubo1966fluctuation}. This can be well handled by introducing the Nos\'e-Hoover thermostat. However, under the Langevin thermostat, this stochastic approach can lead to an artificial heating effect.
In the simulation of heterogeneous systems such as the gas-liquid coexistence, it often requires a larger batch size or decreasing time step size to maintain the equilibrium distribution of the system, 
reducing the efficiency of the RBM \cite{jin2022RBMSOI}.
The artificial heating effect raises the kinetic energy of particles to overcome the intermolecular forces, which may lead to unphysical outcomes in critical issues such as phase transition simulations \cite{guillin2024some}.

In this paper, we present the variance-reduced random batch method for Langevin dynamics. The covariance matrix of the random forces in the drifting term is approximated and used to correct the Brownian motion, achieving a significant variance reduction and leading to the consistency of the single-step fluctuation-dissipation theorem. Theoretically, we demonstrate that the variance reduction method has higher discrete numerical accuracy in the statistical sense. The method is implemented into the Lennard-Jones fluid systems and Coulomb systems, and numerical results demonstrate the attractive performance with a significant reduction in the force variance.
Notably, our algorithm exhibits linear computational complexity and effectively preserves the non-uniform properties observed in strongly correlated systems, thereby holding significant potential for large-scale simulations of heterogeneous systems.

\iffalse
The structure of this paper is organized as follows.
Section \ref{sec:VCRBM} introduces the variance correction random batch method based on the underdamped Langevin dynamics.
In Section \ref{sec:numerical}, we provide numerical comparisons of the Lennard-Jones (LJ) fluid and the implicit colloid-ion solution systems to evaluate the performance of the proposed method.
Finally, in Section \ref{sec:conc}, we present our conclusions, summarizing the key findings and contributions of this study.
\fi

\section{Method}\label{sec:VCRBM}
Consider a system of $N$ particles at position $\bm{r}_i$ for $i=1,2,\cdots, N$ within a simulation box $\Omega$ with periodic boundary conditions.
Suppose that the system is immersed in an infinite isothermal bath, and the particles obey the second-order Langevin dynamics,
\begin{equation}
\label{eq:OrigLangv}
    % \left\{
    % \begin{array}{l}
    % \mathrm{d}\bm{r}_i=\bm{v}_i\mathrm{d}t,\\
    % m_i\mathrm{d}\bm{v}_i=\bm{F}_i\mathrm{d}t-\bm{\Gamma} \bm{v}_i\mathrm{d}t+\bm{D}d\bm{W}_i,
    % \end{array}
    % \right.
    m_i\mathrm{d}^2\bm{r}_i=\bm{F}_i\mathrm{d}t-\bm{\Gamma} \mathrm{d}\bm{r}_i + \bm{D}\mathrm{d}\bm{W}_i
\end{equation}
where $\bm{r}_i$ is the particle position, $\bm{\Gamma}$ is the damping matrix, $\bm{D}$ is the diffusion matrix, and $\{\bm{W}_{i}\}$ is i.i.d. Brownian motion.
The friction term $-\bm{\Gamma} \mathrm{d}\bm{r}_i$ and fluctuation force $\bm{D}\mathrm{d}\bm{W}_i$ satisfy the fluctuation-dissipation theorem \cite{kubo1966fluctuation,marconi2008fluctuation,prost2009generalized}
\begin{equation}
    2k_BT\bm{\Gamma}=\bm{D}^2,
\end{equation}
with $k_B T$ being the thermal energy. Introducing the particle velocity $\bm{v}_i$ with $\mathrm{d}\bm{r}_i=\bm{v}_i\mathrm{d}t$, the classical Langevin dynamics Eq. \eqref{eq:OrigLangv} is discretized by the Euler-Maruyama scheme \cite{kloeden1992stochastic} as follows,
\begin{equation}
\label{eq::CL_EM}
    \left\{
    \begin{array}{l}
    \Delta \bm{r}_i^{n}=\bm{v}_i^{n}\tau\\
    m_i \Delta\bm{v}_i^{n}=\bm{F}_i^{n}\tau-\bm{\Gamma} \bm{v}_i^{n}\tau+\bm{\xi}_i^{n},
    \end{array}
    \right.
\end{equation}
where $\Delta$ is the forward difference operator, $\tau$ is the time step, the superscript $n$ denotes the simulation step at $t=n\tau$, and $\bm{\xi}_i^{n}\sim \mathcal{N}(0,\bm{D}^{2}\tau)$ denotes the discrete Brownian motion.

The total force acting on particle $i$ in Eq. \eqref{eq::CL_EM} is composed of the external force and the interparticle forces,  
\begin{equation}
\bm{F}_i=\bm{F}_i^{\mathrm{ext}}+\sum  \bm{f}_{ij},
\end{equation}
where one has $\bm{f}_{ij}=-\nabla_{\bm{r}_i}U_{ij}$, and the summation runs over all particles in the central and image boxes. Typically, one considers the LJ potential 
\begin{equation}\label{eq:ULJCoul} 
U_{ij}=4\epsilon_{\text{LJ}}\left[\left(\dfrac{\sigma_{\text{LJ}}}{r_{ij}}\right)^{12}-\left(\dfrac{\sigma_{\text{LJ}}}{r_{ij}}\right)^{6}\right], 
\end{equation}
where $\epsilon_{\text{LJ}}$ and $\sigma_{\text{LJ}}$ are the strength and the interaction distance scale.
The classical Langevin dynamics introduces a cutoff radius $r_s$, and particles beyond this distance are ignored in the force calculation. This results in a linear computational complexity by the linked cell list algorithm \cite{yao2004improved,meloni2007efficient,brown2011implementing,welling2011efficiency}, but for simulating the state of the gas-liquid coexistence, a large $r_s$ is required for accurate simulations.
The RBL \cite{liang2021randomRBL} has been proposed to reduce the number of interacting neighbors as well as the computational cost, by introducing the second cutoff radius $r_c$, and dividing the interaction region into a core region $r<r_c$, and a shell region $r_c\leq r<r_s$.
% , schematically depicted in Fig.~\ref{fig:sketch}.
% In the RBL, one introduces the second cutoff radius $r_c$, such that the interaction region is divided into a core region $r<r_c$, and a shell region $r_c\leq r<r_s$.
The force is directly calculated inside the core region, and the scaled force is computed in the shell zone by randomly selecting a batch of $P$ particles with uniform probability.
Denote $N_i$ and $C_i$ as the particle number in the shell region and the set of the $P$ selected particles for the $i$-th particle.
The total force acting on this particle is given by
\begin{equation}\label{eq:randomF}
\widetilde{\bm{F}}_i=\bm{F}_i^{\mathrm{ext}}+\sum_{r_{ij}<r_{\mathrm{c}}} \bm{f}_{ij}+\frac{N_i}{P}\sum_{j\in C_i} \bm{f}_{ij},
\end{equation}
and the average force is subtracted from each particle to conserve the total momentum.

\iffalse
\begin{figure}[htbp]	
	\centering
	\includegraphics[width=0.48\textwidth]{figure/scheme.png}
	\caption{Schematic of the core-shell structure of a moving particle for constructing its neighbor list, where the core and shell regions are determined by two radii $r_c$ and $r_s$.}
	\label{fig:sketch}
\end{figure}
\fi

In the RBL-based Langevin dynamics, the second equation of Eq. \eqref{eq::CL_EM} reads
\begin{equation}
\label{eq::RBL_EM} 
    m_i \Delta\bm{v}_i^{n}=\widetilde{\bm{F}}_i^{n}\tau-\bm{\Gamma} \bm{v}_i^{n}\tau+\bm{\xi}_i^{n}, 
\end{equation}
where the approximation Eq. \eqref{eq:randomF} is used as the total force.
It can be proved that the approximate force $\widetilde{\bm{F}}_i$ is unbiased with a bounded force variance and the convergent distribution is close to the stationary one \cite{liang2021randomRBL}.
The acceleration ratio can be approximated by the ratio of neighbor list size, as $[(4\pi/3) \rho r_\mathrm{s}^3] / [(4\pi/3) \rho r_\mathrm{c}^3 + P]$, where $\rho$ is the particle density.
% One can roughly estimate the acceleration ratio by counting the size of neighbor list. Assuming that the particles have a density $\rho$, one finds that the number of neighbors per particle is reduced from nearly $(4\pi/3) \rho r_\mathrm{s}^3$ to $(4\pi/3) \rho r_\mathrm{c}^3 + P$. This reduction can be significant in many applications.  

% However, the RBL introduces artificial heat due to the stochastic approximation of the force, leading to the $\cO(\tau/P)$ velocity evolution error between Eq. \eqref{eq::CL_EM} and Eq. \eqref{eq::RBL_EM} and $\cO(\sqrt{\tau/P})$ convergence error \cite{liang2021randomRBL} and may destroy the correlation and the architecture of non-uniform distribution between particles.
% This can be well handled by the Nosé-Hoover thermostat \cite{hoover1985canonical} for molecular dynamics at the NVT ensemble \cite{liang2021randomRBL}.
% It becomes nontrivial in the context of the Langevin dynamics, which may destroy the correlation effect and the architecture of non-uniform distribution between particles.

Under the Euler-Maruyama scheme, the mean-field limit of the RBM results in an effective dynamics \cite{inass2023quantifying,guillin2024some}, where
the random force introduces an additional contribution into the Brownian motion, such that $\bm{Z}_t\sim \mathcal{N}\left(0, \bm{D}^{2}\tau+ \bm{\Sigma}\tau^{2}\right)$. Here, $\bm{\Sigma}$ is the force covariance matrix.
Therefore, it is necessary to reduce the force variance to maintain the fluctuation-dissipation theorem. 
One can approximately calculate the covariance matrix of the random force by
\begin{equation}\label{eq:varF}
\begin{array}{c}
\textbf{var}(\widetilde{\bm{F}}_i) \approx \dfrac{N_i^2}{P^2} \left( \sum\limits_{j\in C_P} \bm{f}_{ij}\bm{f}_{ij}^{\mathrm{T}} -P\cdot\overline{\bm{F}}_{i}\overline{\bm{F}}_{i}^{\mathrm{T}} \right),\\
\overline{\bm{F}}_{i}:=\dfrac{1}{P}\sum\limits_{j\in C_P}\bm{f}_{ij}.
\end{array}
\end{equation}
Clearly, it is bounded due to the core-shell structure of the neighbor list and all particles in $C_i$ satisfying $r_{ij}\geq r_{c}$.
One subsequently subtracts the variance term from the Brownian motion in discrete Langevin dynamics with time step $\tau$ such that $\bm{\xi}_i^{n}$ in Eq. \eqref{eq::RBL_EM} is replaced by $\widetilde{\bm{\xi}}_i^{n}$ with
\begin{equation}\label{eq::VRRBL}  
\widetilde{\bm{\xi}}_i^{n} \sim \mathcal{N}\left(0, \bm{D}^{2}\tau -\textbf{var}(\widetilde{\bm{F}}_i^n)\tau^{2}\right), 
\end{equation}
where $\widetilde{\bm{F}}_i^n$ is calculated by \eqref{eq:randomF} at the $n$-th time step.
% This results in a novel Langevin dynamics with the variance reduction technique.
The process of the Langevin dynamics evolved with the Brownian motion \eqref{eq::VRRBL} is presented in Algorithm \ref{al::VR-RBL}.
% For the LJ system with the RBL for the neighbor list, the process of performing the VR-RBL based Langevin simulations is present in Algorithm \ref{al::VR-RBL}.
\begin{algorithm}[H]
	\caption{Varaince-reduced random batch Langevin dynamics}\label{al::VR-RBL} % Langevin dynamics for LJ systems)
	\begin{algorithmic}[1]
		\STATE Choose $r_s$ (the cutoff radius), $r_c$ (the core radius), $\tau$ (time step), $N_{t}$ (total simulation steps), and $P$ (batch size).  Initialize positions, velocities and charges of all particles.
		\FOR {$n \text{ in } 1: N_{t}$}
		\STATE \quad Create the cell lists.
            \STATE \quad For each particle $i$, select $P$ particles randomly into $C_i$ from its shell zone if $N_i>P$, otherwise select all particles. 
            \STATE \quad  For each particle $i$, calculate the selected pair interactions $\bm{f}_{ij}$ for all $r_{ij}<r_c$ and $j\in C_P$, $j\neq i$, then calculate stochastic force $\widetilde{\bm{F}}_i$ by \eqref{eq:randomF} and covariance matrix $\textbf{var}(\widetilde{\bm{F}}_i)$ by \eqref{eq:varF}.
            \STATE \quad  Integrate the Langevin dynamics \eqref{eq::RBL_EM} with the Brownian motion \eqref{eq::VRRBL}.
		\ENDFOR
	\end{algorithmic}
\end{algorithm}

Considering the estimate of the local truncation error, one omits the superscript $n$ for simplicity.
    The velocity evolution of a single step in the RBL method is provided by \eqref{eq::RBL_EM}, where $\widetilde{\bm{F}}_i$ is the unbiased estimation of $\bm{F}_i$, hence it can be approximately regarded as $\widetilde{\bm{F}}_i = \bm{F}_i +\bm{\eta}_i$ with $\bm{\eta}_i\sim \mathcal{N}(0,\bm{\mathcal{V}})$ and $\bm{\mathcal{V}}$ denoting the covariance matrix of $\widetilde{\bm{F}}_i$. 
    The discrete scheme can be approximated as
    \begin{equation}
    \label{eq::RBL_LE}
    m_i \Delta \bm{v}_i=(\bm{F}_i-\bm{\Gamma} \bm{v}_i)\tau+ \hat{\bm{\xi}_i},
    \end{equation}
    where $\hat{\bm{\xi}_i}= \bm{\xi}_i + \bm{\eta}_i  \tau + \cO(\tau^{3/2})$ denotes the noise term. Hence the local truncation error of the RBL discretization is $\cO(\tau)$, which leads to the inconsistency with the fluctuation-dissipation theorem. As for the proposed dynamics with variance reduction technique, Theorem~\ref{thm::VR} demonstrates that it reduces the local truncation error to $\cO(\tau^{3/2})$.

\begin{theorem}
\label{thm::VR}
Given the force variance, the Langevin dynamics evolved with the random force \eqref{eq:varF} and the Brownian motion \eqref{eq::VRRBL} exhibits a local truncation error of $\cO(\tau^{3/2})$ in the sense of statistical expectation.
\end{theorem}
\begin{proof}
   When applying the variance-reduced Brownian motion $\widetilde{\bm{\xi}}_i$ in \eqref{eq::VRRBL} instead of $\bm{\xi}_i$ in the RBL dynamics \eqref{eq::RBL_EM}, one finds that the noise term $\hat{\bm{\xi}}_i$  becomes 
    \begin{equation}
    \label{eq::Taylor}
    \begin{aligned} 
    \hat{\bm{\xi}}_i&=\bm{\eta}_i \tau + \widetilde{\bm{\xi}}_i +\cO(\tau^{3/2}) \\
    &= \bm{\gamma}_i+\cO(\tau^{3/2}). 
    \end{aligned}
    \end{equation}
where $\bm{\gamma}_i\sim \mathcal{N}(0, \bm{D}^2\tau+(\bm{\mathcal{V}}-\textbf{var}(\widetilde{\bm{F}}_i))\tau^2)$.
% where $\widetilde{\bm{\eta}}_i\sim \mathcal{N}(0, \bm{\mathcal{V}}\tau^{2} -\textbf{var}(\widetilde{\bm{F}}_i)\tau^2)$. 
The difference between $\bm{\gamma}_i$ and $\bm{\xi}_i$ is of $\cO(\tau^{3/2})$, since $\textbf{var}(\widetilde{\bm{F}}_i)$ is an unbiased estimation of $\bm{\mathcal{V}}$. So $\hat{\bm{\xi}}_i=\bm{\xi}_i +\cO(\tau^{3/2})$ shows the higher order truncation error of the proposed dynamics. 
\end{proof}

Along with insights from effective dynamics \cite{inass2023quantifying,guillin2024some}, one also  obtains that the VR-RBL is in consistent with the fluctuation-dissipation theorem since it holds
\begin{equation}
    2k_BT\bm{\Gamma} = \bm{D}^{2} + (\bm{\mathcal{V}}-\textbf{var}(\widetilde{\bm{F}}_i))\tau + \cO(\tau^{3/2}).
\end{equation} 
Since the proposed dynamics have a higher order local truncation error, one can select a smaller batch size $P$ to achieve accurate simulation results.
The primary computational complexity of variance-reduced dynamics centers around the force calculation of particles, while the computation of the variance reduction term entails comparatively lower computing overhead, requiring only the summation process of stored force arrays.
The variance-reduced method is demonstrated to be accurate through numerical results of LJ fluid and colloid-ion systems in section \ref{sec:numerical}, indicating its potential as a general technique for the RBM that sheds new light on precise and effective Langevin dynamics for various applications.

\section{Results}\label{sec:numerical}

We perform simulations to show the effectiveness of the proposed Langevin dynamics method by comparing the results of the three methods: classical Langevin dynamics (CL), the random batch list method without variance reduction (RBL), and the variance-reduced method (VR-RBL). The CL is a brute-force method and is supposed to be accurate as the reference solution.  All the methods discretize the Langevin dynamics by the Euler-Maruyama scheme. We study two benchmark systems:
one is the LJ fluid exhibiting a gas-liquid coexistence state, and the second is a colloid-ion system with varying correlation coefficients \cite{srinivas2004self,beck2004methods,miao2015accelerated}.
The simulations are conducted on a Linux system equipped with an Intel Xeon Scalable Cascade Lake processor running at 2.5 GHz, utilizing 1 single CPU core and 4 GB of memory.

\begin{figure*}[t!]
	\centering
	\includegraphics[width=0.48\textwidth]{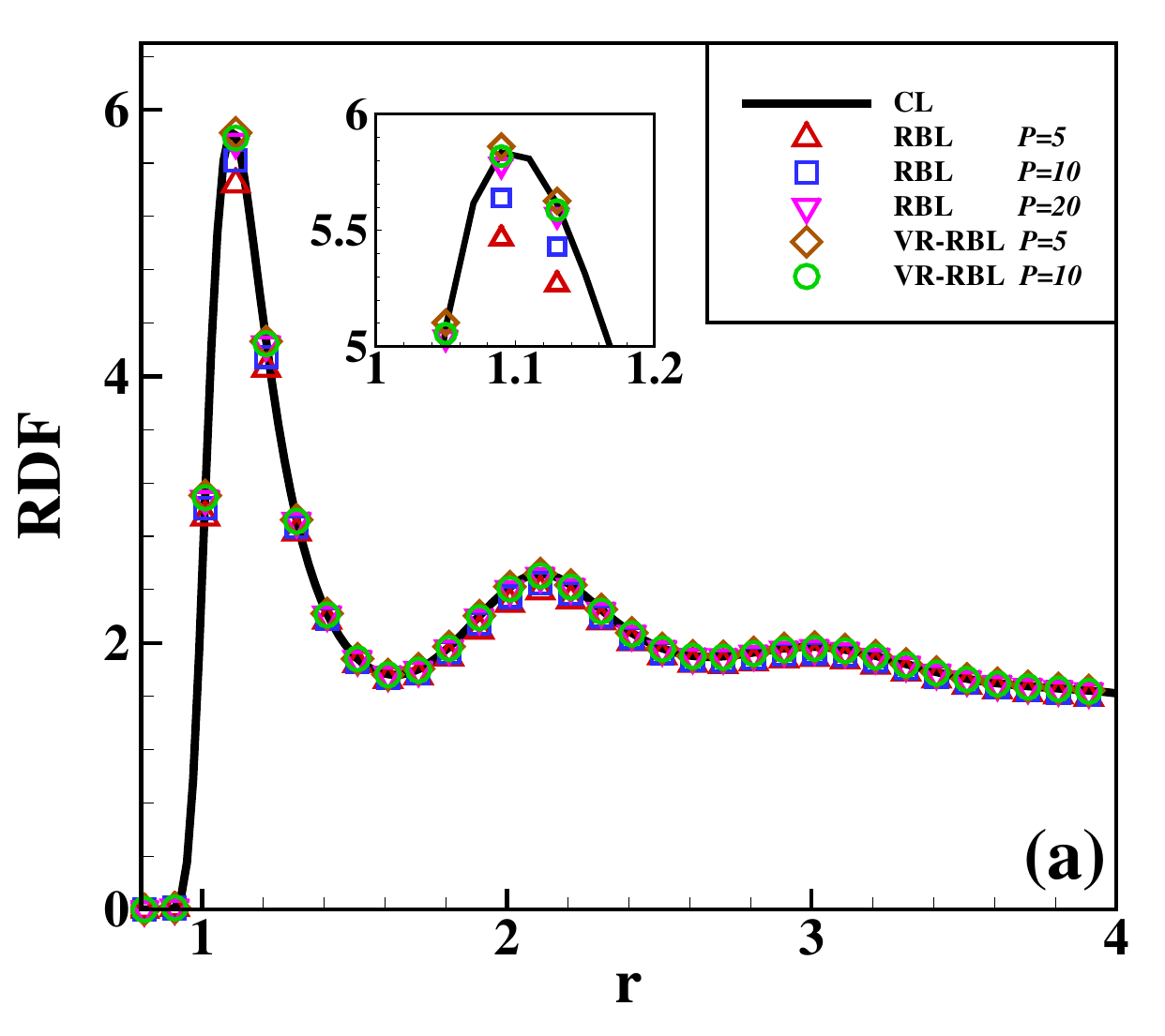}
	\includegraphics[width=0.48\textwidth]{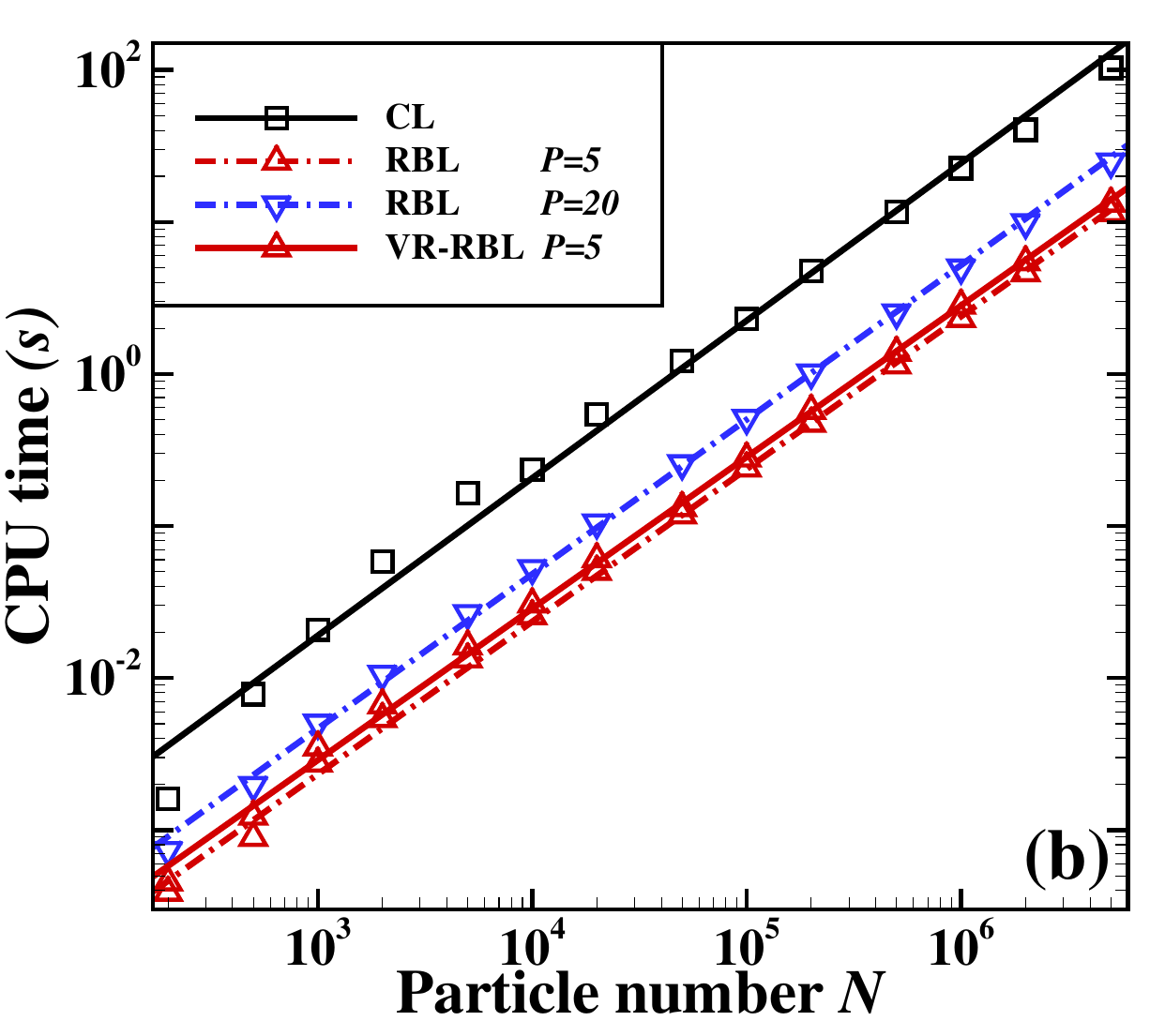}
	\caption{The RDF (a) and CPU time per step (b) calculated by the CL, RBL and VR-RBL methods. The reduced units are applied, and the lines in (b) indicate the linear fitting of all methods.}
	\label{fig:LJRDF}
\end{figure*}

\subsection{Lennard-Jones fluids}

Consider the LJ fluid with interparticle potential given in \eqref{eq:ULJCoul} and parameters $\sigma_{\text{LJ}}=\epsilon_{\text{LJ}}=1$. One sets the particle number $N=500$ in a cubic box such that the average particle density $\rho=0.2$.  
The reduced units with $k_B=1$, $\bm{\Gamma} = \bm{I}$ and time step $\tau=0.01$ are applied. The temperature takes $T=0.9$ at which the system exhibits gas-liquid coexistence state. The cutoff radius takes $r_\mathrm{s}=6.0$ such that the CL remains accurate at the gas-liquid coexistence state, and $r_\mathrm{c}=2.0$. In the calculation, one performs $10^5$ time steps for equilibrium and then $2\times 10^5$ steps for statistics. Fig.~\ref{fig:LJRDF}(a) presents the radial distribution functions (RDF) by the three methods, where one calculates $P=5$, $10$, and $20$ for the RBL, and $P=5$ and $10$ for the VR-RBL. One can see that the variance reduction significantly improves the accuracy for given $P$. The VR-RBL converges to the CL with the small batch size $P=5$. In comparison, the RBL achieves the same accuracy at $P=20$. 
In Fig.~\ref{fig:LJRDF}(b), we present CPU performance for the three methods by using the same setup with fixed density, but the number of particles increases up to $N=5\times 10^6$. The computational cost of all methods increases linearly with the number of particles, but one can observe a significant reduction in the prefactor of the scaling law by imposing random batch techniques.
One can roughly estimate that, for the LJ fluid system of $\rho=0.2$, each particle has $125$ neighboring particles within $r_\mathrm{s}=6.0$ with $11$ neighbors in the core region.
Fig.~\ref{fig:LJRDF}(b) shows that the RBL with $P=20$ is $4.6$ times faster than the CL, and for the VR-RBL the accelerating rate becomes $7.9$. These results are in agreement with the rough estimate of neighboring number and demonstrate the high efficiency of the VR-RBL given the accuracy level.

\subsection{Colloid-ion Systems}

The second example is a colloid-ion system in a spherical simulation volume of radius $R=20$. A colloidal particle with charge $Q_0=-20$ and radius $r_0=5$ is placed at the center, along with $100$ monovalent cations and $80$ monovalent anions with the same radius $r\equiv 0.5$ in the remaining volume. All quantities are provided in reduced units. The interparticle potential is composed of the LJ potential in \eqref{eq:ULJCoul} of parameters $\sigma_{\text{LJ}}=\epsilon_{\text{LJ}}=1$, and  Coulomb interaction $U_{ij}^{\mathrm{Coul}}=\alpha  q_i q_j/r_{ij}$ with $\alpha$ being the coupling strength. 
A spherical Wigner-Seitz (WS) cell model \cite{groot1991ion,tamashiro1998donnan} is employed for the boundary condition in our simulation with the LJ for the ion-wall interaction. 
One sets the temperature $T=1.0$, the unit damping matrix $\bm{\Gamma}=\bm{I}$ and time step $\tau=0.005$ in the CL, RBL and VR-RBL methods.
In the simulation, one performs $5\times 10^{6}$ steps for equilibrium and another $5\times 10^{6}$ steps for statistics.  The coupling strength takes $\alpha=2.0$ and $10.0$, representing a weak and strong coupling system with corresponding cutoff radii setting $r_\mathrm{c}=4.0$ and $5.0$, respectively. Due to the Coulomb interaction, the results of CL are derived from calculating all particle interactions. % , i.e., one has $r_\mathrm{s}=R$.

Fig.~\ref{fig:IntQ} presents the results of the charge density of the cation (normalized by its bulk density) and the integrated charge distribution along the radial direction. The integrated charge distribution is calculated through $Q_{\mathrm{int}}(r)=Q_0+\sum\int_{0}^{r} q_i\rho_i(s)\mathrm{d}s$, where $\rho_i(s)$ is calculated by a statistical average. One can observe that RBL and VR-RBL methods are accurate at the weak-coupling case of $\alpha=2.0$ (Fig.~\ref{fig:IntQ}(a,c)). However, the RBL with a small batch size fails to capture the stationary distribution in the $\alpha=10.0$ case, which has strong Coulomb interactions (Fig.~\ref{fig:IntQ}(b,d)). The VR-RBL with $P=5$ has almost the same accuracy as the RBL with $P= 20$, in agreement with the CL result. This result reveals the accuracy and efficiency advantages of the VR-RBL method.

\begin{figure*}[htbp]	
	\centering
	\includegraphics[width=0.48\textwidth]{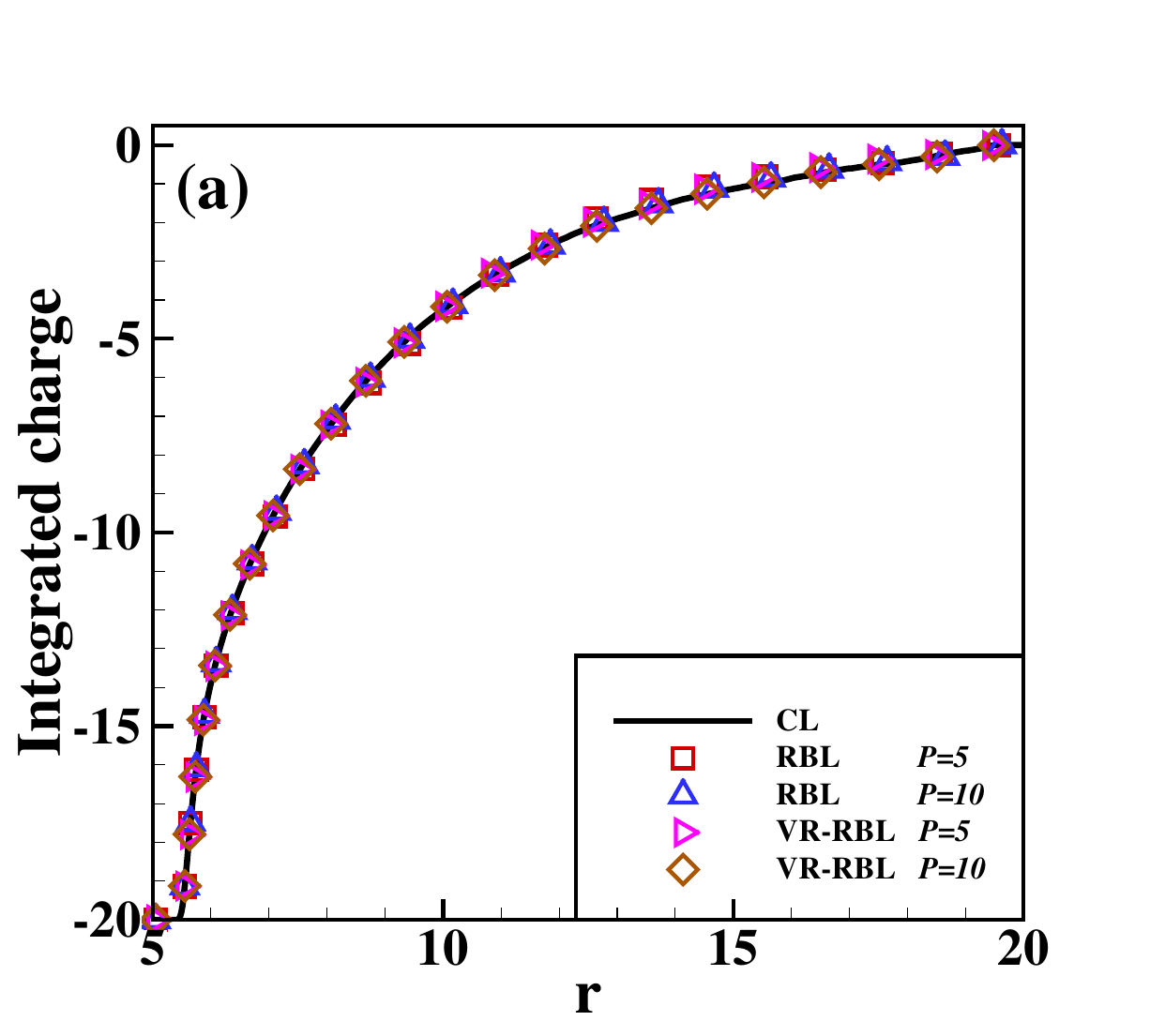}
	\includegraphics[width=0.48\textwidth]{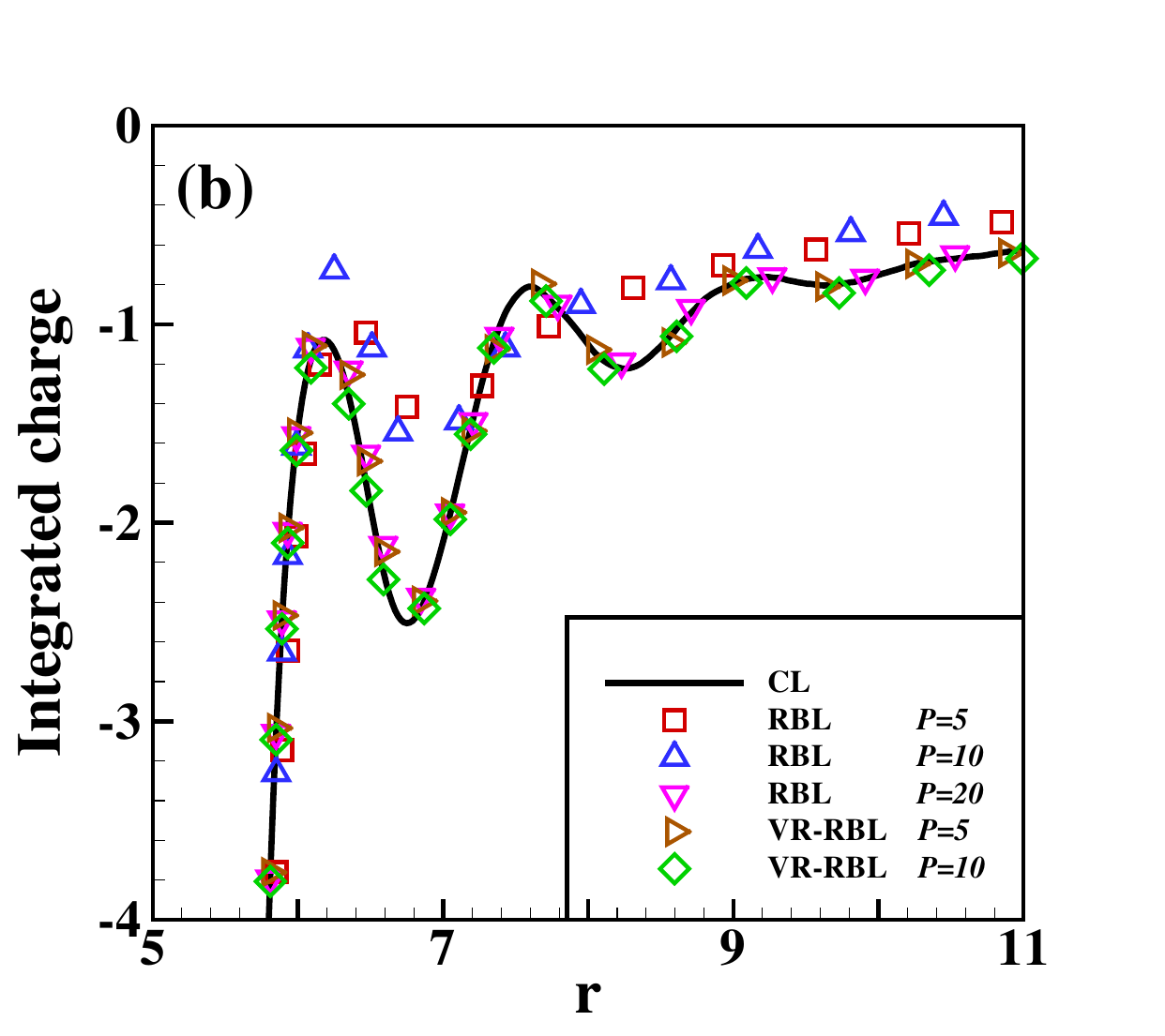}
	\includegraphics[width=0.48\textwidth]{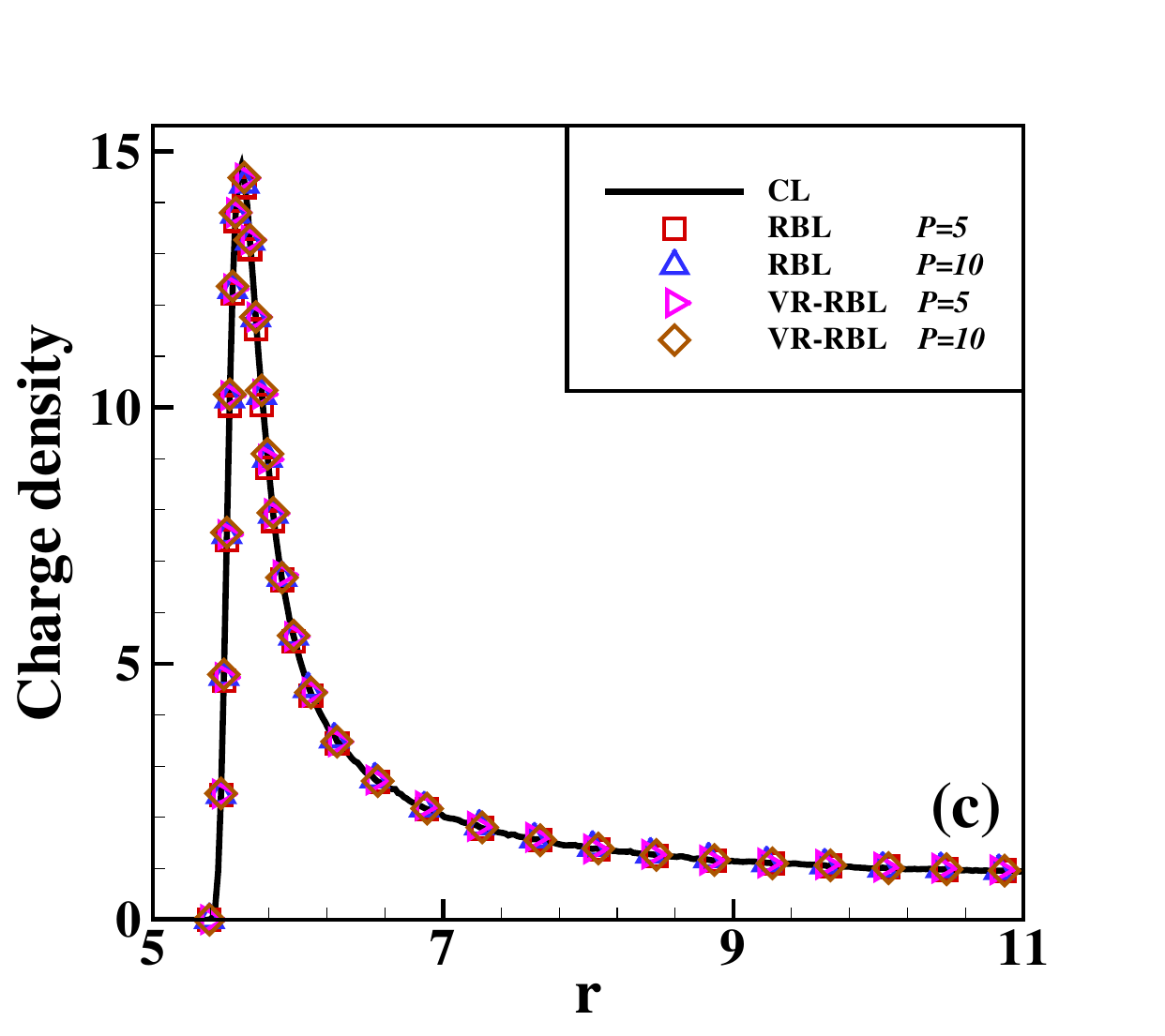}
	\includegraphics[width=0.48\textwidth]{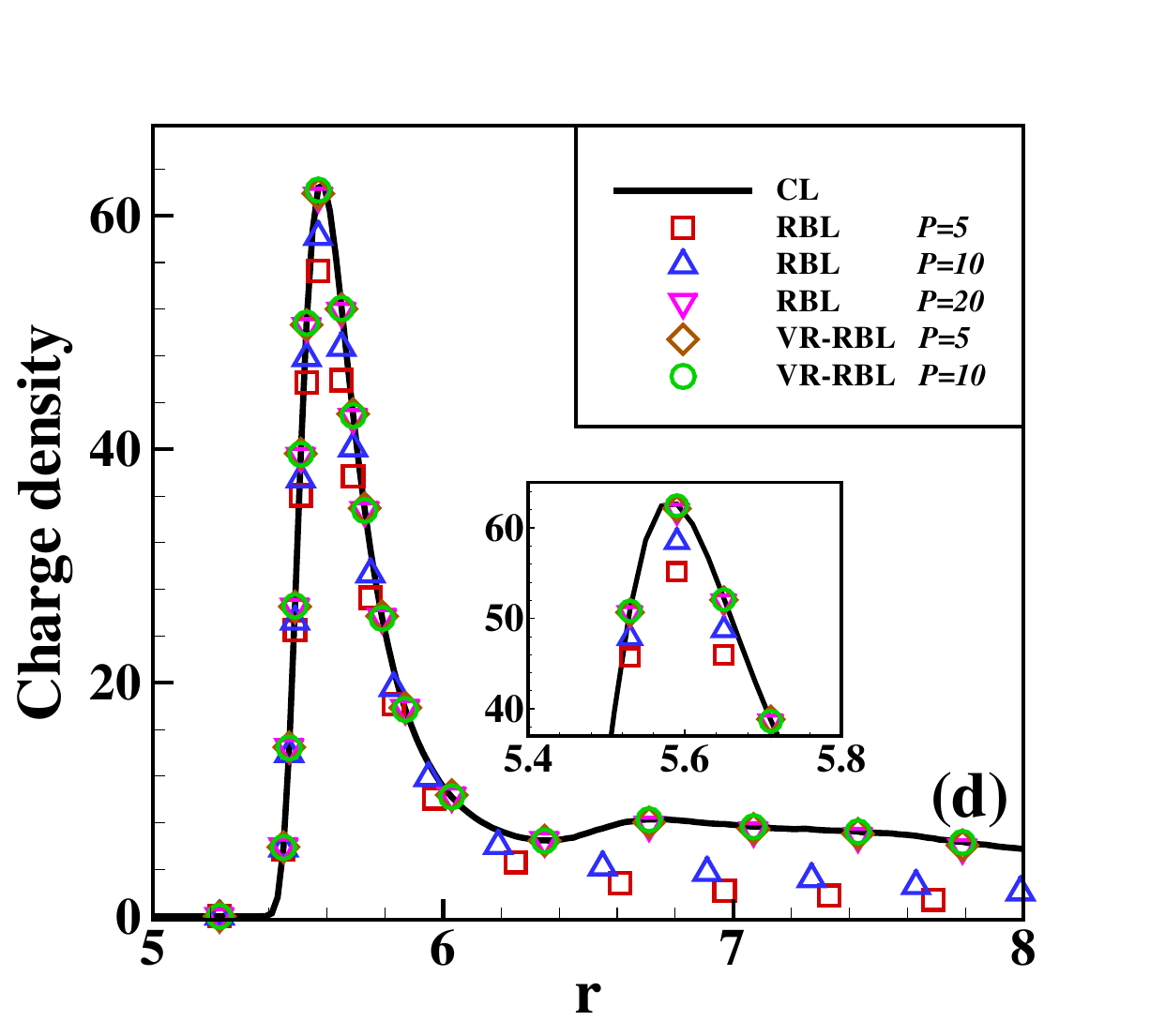}
	\caption{Integrated charge and charge density of cation simulated by the CL, RBL and VR-RBL methods for coupling strength $\alpha=2.0$ [(a,c)] and $\alpha=10.0$ [(b,d)]. All quantities are provided in reduced units.}
	\label{fig:IntQ}
\end{figure*}

It is also interesting to investigate how the time step affects the simulation accuracy. One considers the RBL and VR-RBL for the $\alpha=10.0$ case simulating the same times for the equilibrium and statistics, with the time step ranging from $\tau=0.001$ to $0.02$. Fig.~\ref{fig:RE} presents the relative errors of the kinetic energies and potential energies in the system. These errors are calculated by statistical averages of the energies compared with those of the CL. 
One can clearly observe the linear convergence of the VR-RBL with $1/P$ for both the kinetic and potential energies. Both methods exhibit a convergence rate of $\mathcal{O}(\sqrt{\tau})$; however, the prefactor for the VR-RBL is significantly smaller than that of the RBL. Moreover, the VR-RBL method can accurately capture the property of the system with strong coupling strength even at larger $\tau$, further highlighting its application potential.

\begin{figure*}[htbp]
	\centering
	\includegraphics[width=0.48\textwidth]{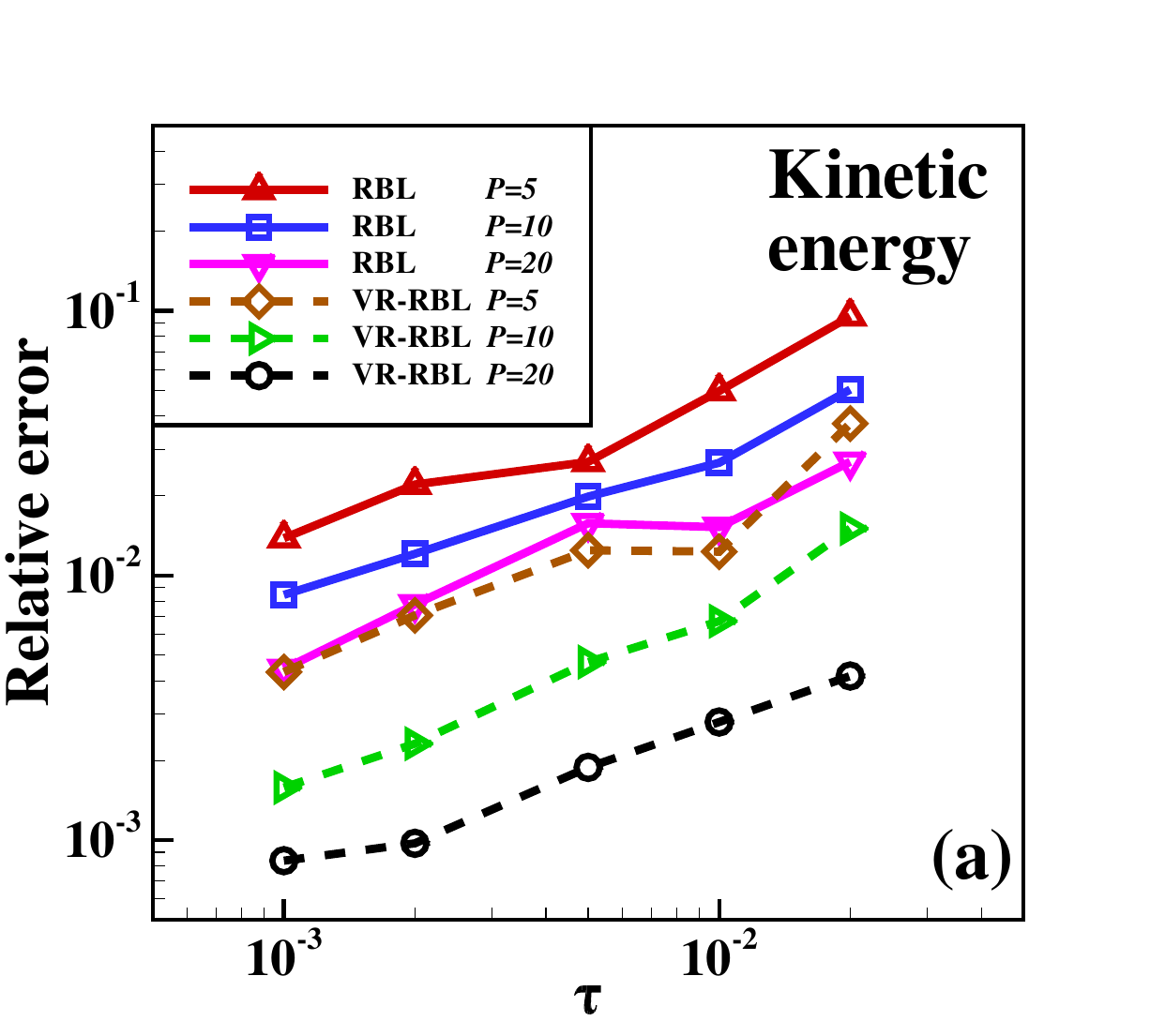}
	\includegraphics[width=0.48\textwidth]{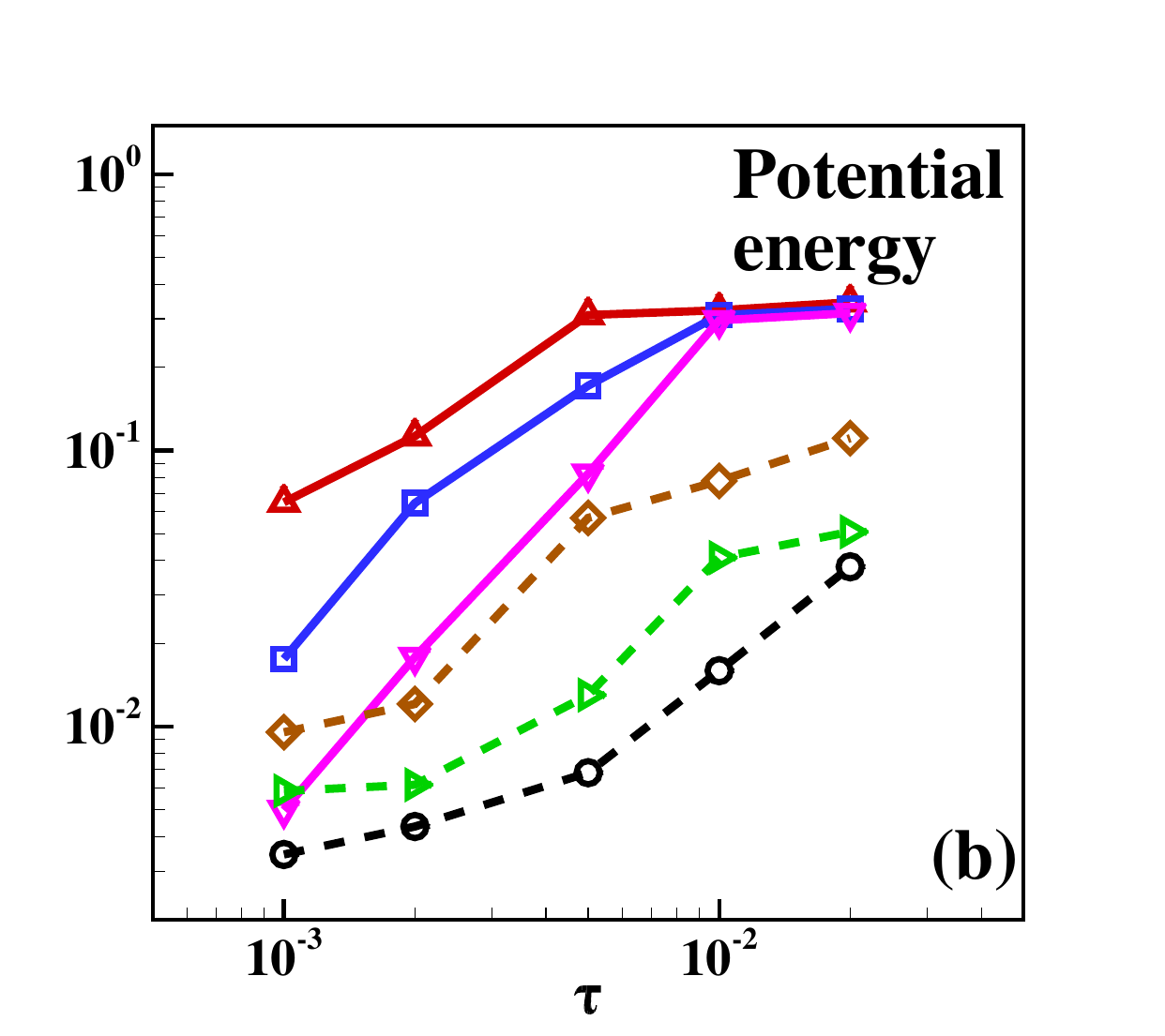}
	\caption{The relative error of kinetic energy(a) and potential energy(b) calculated by the RBL and VR-RBL methods for strong Coulomb correlation system $\alpha=10.0$ with different time steps $\tau$. }
	\label{fig:RE}
\end{figure*}

\section{Conclusions} \label{sec:conc}

In summary, we have proposed a variance-reduced random batch Langevin dynamics for the efficient and accurate simulations of particle systems. The VR-RBL computes interactions using the random batch method while simultaneously estimating the covariance matrix of the forces. The covariance matrix is used to correct the Brownian motion with little additional computational cost to mitigate artificial heating effects and maintain the consistency with the fluctuation-dissipation theorem. Numerical experiments conducted on LJ fluids with gas-liquid coexistence state and colloid-ion systems with different Coulomb strengths demonstrate the accuracy and computational efficiency of the VR-RBL. It is promising for simulations of large-scale heterogeneous systems, allowing for the use of larger time steps to reduce computational time. The rigorous analysis on the convergence of the VR-RBL, along with strategies for further optimizing its accuracy and computational efficiency, will be studied in our subsequent work. It is noteworthy that, the extension of our proposed dynamics to other random batch processes, such as the random batch Ewald and the random batch sum-of-Gaussians methods for long-range interaction systems in molecular dynamics simulations \cite{JinLiXuZhao2020,liang2023random}, is straightforward to implement.

% In summary, we propose a variance-reduced random batch list method for sampling from the stationary distribution of many-body particle systems with singular kernels using Langevin dynamics.
% The VR-RBL method combines the random batch list method with the variance reduction technique to reduce computational complexity, eliminate the artificial heating effect, \tcb{and better maintain the fluctuation-dissipation theorem}.
% The VR-RBL method is proven to achieve higher discrete numerical accuracy.
% The simulations on LJ fluid and colloid-ion systems are conducted to demonstrate the accuracy and efficiency of the VR-RBL algorithm.
% The VR-RBL method can guarantee the inhomogeneous nature of particle distributions in multinomial coexistence systems or strongly correlated interaction systems.

% The VR-RBL method shows promise for implementations in all-atom simulations of real physical systems when using the random mini-batch idea, such as the RBL method.
% However, further research is needed to rigorously analyze the convergence of the VR-RBL method, as remaining an open problem.
% Additionally, optimizing the computational efficiency and accuracy of the VR-RBL method, exploring potential modifications or extensions, and validating its effectiveness in practical applications related to real-world physical systems are interesting avenues for future investigations.

\section*{Acknowledgement}

This work is supported by the National Natural Science Foundation of China (grant No. 12325113) and the Science and Technology Commission of Shanghai Municipality (grant Nos. 23JC1402300 and 21JC1403700).
The authors also acknowledge the support from the HPC center of Shanghai Jiao Tong University.

\section*{Conflict of interest}

The authors declare that they have no conflict of interest.

\section*{Data Availability Statement}

The data that support the findings of this study are available from the corresponding author upon reasonable request.

\bibliographystyle{elsart-num} %{elsart-num-sort} {plain}  {siam}%
\bibliography{groupbib,ewald}

\end{document}